\documentclass[letterpaper,10pt,conference]{ieeeconf}  

\IEEEoverridecommandlockouts                              
\usepackage{graphics} 
\usepackage{amsmath} 
\usepackage{amssymb}  
\usepackage{amsfonts}
\usepackage[english]{babel}
\usepackage[T1]{fontenc}
\usepackage{comment}
\usepackage{mathtools}
\usepackage{multirow}
\usepackage{cite}
\usepackage{tikz}
\usepackage{booktabs} 
\usepackage{tabularx}
\usepackage{url}

\newcommand{\dd}{{\mathsf{d}}}

\newcommand{\R}{\ensuremath{{\mathbb R}}}

\newcommand{\KK}{{\mathcal K}}

\newtheorem{lem}{{Lemma}}

\newtheorem{thm}{{Theorem}}
\newtheorem{asm}{{Assumption}}
\newtheorem{rem}{{Remark}}
\newtheorem{exam}{{Example}}

\title{\LARGE \bf Model Reference Gaussian Process Regression:\\ Data-Driven Output Feedback Controller}

\author{Hyuntae Kim, Hamin Chang, and Hyungbo Shim
\thanks{This work was supported by the grant from Hyundai Motor Company's R\&D Division.}
\thanks{All authors are with ASRI, Department of Electrical and Computer Engineering,
Seoul National University, 1 Gwanak-ro, Gwanak-gu, Seoul, 08826, Korea. Corresponding author: {\tt\small hshim@snu.ac.kr}}%
}

\begin{document}

\maketitle
\thispagestyle{plain}
\pagestyle{plain}

\begin{abstract}

Data-driven controls using Gaussian process regression have recently gained much attention. 
In such approaches, system identification by Gaussian process regression is mostly followed by model-based controller designs.
However, the outcomes of Gaussian process regression are often too complicated to apply conventional control designs, which makes the numerical design such as model predictive control employed in many cases.
To overcome the restriction, our idea is to perform Gaussian process regression to the inverse of the plant with the same input/output data for the conventional regression.
With the inverse, one can design a model reference controller without resorting to numerical control methods.
This paper considers single-input single-output (SISO) discrete-time nonlinear systems of minimum phase with relative degree one. 
It is highlighted that the model reference Gaussian process regression controller is designed directly from pre-collected input/output data without system identification. 
\end{abstract}

\section{Introduction}

Gaussian process regression (GPR) \cite{GPML}, one of the most well-known regression tools for nonlinear functions, has been extensively used in various fields by virtue of the following properties \cite{KocijanBook}.
First, since it is a nonparametric method, it has some flexibility to deal with a large amount of data.
Secondly, prior knowledge of the regression target can easily be incorporated. 
Finally, it gives some confidence information about the regression result, which can be utilized to measure the regression error.

Particularly in control systems, GPR has been mainly applied for identifying unknown nonlinear systems using input/output or even state data before designing a model-based controller for the identified model. For instance, \cite{RDCA03} and \cite{JRCA04} show that a model predictive controller can be designed based on the model identified by GPR.
Moreover, its real world applications are presented in \cite{GEF16} and \cite{CATJ16} for quadrotors and mobile robots, respectively.
In addition, combining the prior knowledge of a nominal model, 
\cite{LJM19} and \cite{JLAM19} demonstrate the utility of such Gaussian process-based model predictive control (GP-MPC) method in autonomous racing systems by identifying a residual model instead of the full dynamical system.
Also, \cite{Scaramuzza} presents real world experiments of quadrotors controlled by the GP-MPC approach, where only aerodynamic effects on quadrotors are modeled by Gaussian process.
On the other hand, \cite{Umlauft17} and \cite{Umlauft18} propose a feedback linearization controller for the system, which is identified by GPR. They also provide Lyapunov stability analysis of the controlled system based on the result concerning the error of the identification in \cite{12TIT}. 
On the other hand, \cite{Umlauft20} and \cite{Jiao22} propose event triggered online learning of GPR in order to increase data efficiency.

However, most of these studies focus only on the system identification capability of the GPR in the sense that controllers should be designed only after the system identification by using the GPR is completed.
This leads to a problem with controller design because even if the given system has a relatively simple analytic formula, its identified model by the GPR can be too complicated consisting of a summation of as many terms as the number of data points.
This complexity has often restricted applicable control methods, so that numerical methods such as MPC (model predictive control) are typically employed for the identified model.
Since numerical control methods require a certain amount of online computation resource, its utility can sometimes be limited.

To overcome the restriction, we propose identification of the inverse of the given system by GPR, with the purpose of using it for model reference control.
In this way, we can bypass numerical controls and can combine classical controls resulting in a data-driven controller, which we call model reference Gaussian process regression (MR-GPR) control in this paper.
Since it is natural to assume that we have access to only input/output measurements of the plant, we propose the MR-GPR controller in the form of an output feedback control.
Therefore, the GPR is performed only with input/output data of the system. 
Since our approach is based on input/output inversion in some sense, a few limitations naturally follow such as causality and minimum phase issues.
In this paper, we assume that the system has relative degree one to resolve the causality issue, which is not very restrictive because a sampled-data system of a continuous-time system generically has relative degree one.
Moreover, we assume that the system is of minimum phase.

This paper is organized as follows. 
The problem formulation with a class of nonlinear systems under consideration and a couple of assumptions on the class of systems are in Section \ref{sec:problem}. 
In Section \ref{main}, we propose the data-driven MR-GPR controller and explain how to design it using the GPR. Also, a stability analysis of the closed-loop system with the MR-GPR controller is presented. An illustrative example that demonstrates the usefulness of the MR-GPR controller is given in Section \ref{sec:ex}. Finally, this paper is summarized and concluded in Section \ref{sec:conc}.

{\it Notation:}
For column vectors $a$ and $b$,  $[a;b]$ denotes $[a^T,b^T]^T$. 
For discrete-time vector sequences $y(t)$ and $z(t)$,
we define a vector
$$z_{[k,k+T]} := [z(k); z(k+1); \cdots; z(k+T)],$$
and a set
\begin{align*}
    &\{(y(t),z(t))\}^{k+T}_{t=k}\\
    &\quad\quad:= \{(y(k),z(k)),\cdots,(y(k+T),z(k+T))\}.
\end{align*}

\section{Problem Formulation}\label{sec:problem}

Consider a single-input single-output (SISO) nonlinear discrete-time control-affine system with relative degree one in Byrnes-Isidori normal form \cite{Normalform}:
\begin{subequations}\label{SYSTEM}
\begin{align}
y(t+1) &= f(z(t),y(t)) + g(z(t),y(t))u(t) \label{SYSTEM1} \\
z(t+1) &= h( z(t), y(t)) \label{SYSTEM2} 
\end{align}
\end{subequations}
where $u(t)\in\mathbb{R}$ is the input, $z(t)\in\mathbb{R}^{n-1}$ is the state of the zero dynamics, and $y(t)\in\mathbb{R}$ is the output.
It is assumed that the functions $f(\cdot,\cdot)$, $g(\cdot,\cdot)$, and $h(\cdot,\cdot)$ are unknown and only the input/output of the system are available as measurements. 
Also, we assume that the functions {$f(\cdot,\cdot)$}, {$g(\cdot,\cdot)$}, and {$h(\cdot,\cdot)$} are smooth.
In addition, the following assumption is given.

\begin{asm}\label{ASM:1}
	The system \eqref{SYSTEM} satisfies the followings:
	\begin{itemize}
        \item[(a)] The system has global relative degree one, or equivalently, $g(z,y) \neq 0$ for all $(z,y) \in \mathbb{R}^{n}$. 
        Also, the system dimension $n$ and the global relative degree one are known.
        \item[(b)] The internal dynamics \eqref{SYSTEM2} is input-to-state stable with the input being $y$. $\hfill\Box$
  	\end{itemize}  
\end{asm}

If the plant to be controlled is a continuous-time physical system, then its discretization generically yields a discrete-time system of relative degree one \cite{Goodwin}.
Therefore, the system description of \eqref{SYSTEM} may not be too restrictive.
Now, we assume observability of the system (i.e., observability for the state $z$) as follows.

\begin{asm}\label{ASM:OBS}
There exists a smooth mapping $\mathcal{O}: \mathbb{R}^{2n-1} \to \mathbb{R}^{n-1}$ that determines the state $z(t)$ as
\begin{align*}
z(t) = \mathcal{O}([y_{[t,t+n-1]};u_{[t,t+n-2]}])
\end{align*}
for any pair of input $u_{[t,t+n-2]}$ and output $y_{[t,t+n-1]}$ of the system $\eqref{SYSTEM}$.
$\hfill\Box$
\end{asm}

\begin{exam}
For simplicity, let us write $y(t)$ by $y_t$ in this example.
When the system \eqref{SYSTEM} has the form of
\begin{align}\label{eq:exsys}
\begin{split}
y_{t+1} &= f_z(z_t)+f_y(y_t) + u_t \\
z_{t+1} &= h(z_t, y_t)
\end{split}
\end{align}
then Assumption \ref{ASM:OBS} holds if, for any input/output trajectory $u_{[t,t+n-2]}$ and $y_{[t,t+n-1]}$ of \eqref{SYSTEM}, there exists a unique solution $z^* \in \R^{n-1}$ to the equations
\begin{align*}
f_z(z^*) &= y_{t+1} - f_y(y_t) - u_t, \\
f_z(h(z^*,y_t)) &= y_{t+2} - f_y(y_{t+1}) - u_{t+1}, \\
f_z(h(h(z^*,y_t),&f_y(y_t)+f_z(z^*)+u_t)) \\
&= y_{t+3} - f_y(y_{t+2}) - u_{t+2}, \\
& \vdots \\
f_z(h(\cdots(h(z^*,y_t),&f_y(y_t)+f_z(z^*)+u_t), \cdots)) \\
&= y_{t+n-1} - f_y(y_{t+n-2}) - u_{t+n-2}
\end{align*}
which is derived directly from the system \eqref{eq:exsys}.
In this case, $z(t) = z^*$.
$\hfill\Box$
\end{exam}

On the other hand, let us consider a stable reference model given by
\begin{equation}\label{eq:real_closed}
y_r(t+1) = f_r(y_r(t))	\quad \in \R
\end{equation}
which satisfies the additional assumption that
$$y_r(t+1) = f_r(y_r(t)) + \eta(t)$$
is input-to-state stable when $\eta$ is viewed as an input.
In order to make the controlled system \eqref{SYSTEM} become the reference model \eqref{eq:real_closed}, the controller should be
\begin{align}\label{idealstate}
u(t) = \frac{f_r(y(t)) - f(z(t),y(t))}{g(z(t),y(t))}.
\end{align}
For designing the controller \eqref{idealstate}, however, not only the functions $f(\cdot,\cdot)$ and $g(\cdot,\cdot)$ are needed, but also the state $z(t)$ needs to be measured.
In this paper, we present a method to construct the controller \eqref{idealstate} by using only the input/output data of the system \eqref{SYSTEM}.

\section{Main Result}\label{main}

In this section, we design a data-driven controller that can produce almost the same control input as \eqref{idealstate} by using GPR trained by input/output data of the system \eqref{SYSTEM}.

We firstly show that the state $z(t)$ can be expressed by the input/output history of the system \eqref{SYSTEM}.
For this, let
\begin{align}\label{def:zeta}
  \zeta_0(t) := [y_{[t-n+1,t-1]};u_{[t-n+1,t-1]}] \quad \in \R^{2(n-1)}.
\end{align}

\begin{lem}\label{LEM:z}
Under Assumption \ref{ASM:OBS}, there exists a smooth function $\theta:\mathbb{R}^{2(n-1)}\times\mathbb{R} \to \mathbb{R}^{n-1}$, such that the state $z(t)$ of \eqref{SYSTEM} is given by
\begin{align*}
    z(t) = \theta(\zeta_0(t),y(t))
\end{align*}
for all time step $t$.
\end{lem}

\begin{proof}
Since there exists a smooth mapping $\mathcal{O}$ such that $$z(t-n+1) = \mathcal{O}([y_{[t-n+1,t]};u_{[t-n+1,t-1]}])$$ by Assumption~\ref{ASM:OBS}, it follows that
\begin{align*}
z(t) &= h( z(t-1), y(t-1)) \\
&= h( h( z(t-2), y(t-2)), y(t-1)) \\
&= h( h( h( z(t-3), y(t-3)), y(t-2)), y(t-1)) \\
& ~~~~~~~~~~\vdots \\
&= h(\cdots(h(z(t-n+1),y(t-n+1)),\cdots),y(t-1)) \\
&= h(\cdots(h(\mathcal{O}([y_{[t-n+1,t]};u_{[t-n+1,t-1]}]), \\
&~~~~~~~~~~~~~~~~~~~~~~~~~~ y(t-n+1)),\cdots),y(t-1)) \\
&=: \theta(\zeta_0(t),y(t))
\end{align*}
which completes the proof.
\end{proof}

Let us define the vectors 
\begin{align*}
\zeta_1(t) &:= [\zeta_0(t); y(t)] \quad \in \R^{2n-1} \\
\xi(t) &:= [\zeta_1(t); y(t+1)] \quad \in \R^{2n}
\end{align*}
which are composed of an arbitrary input/output trajectory of the system \eqref{SYSTEM}, and suppose that the set $\mathcal{E}$ contains all possible $\xi(t)$.

Define $c:\R^{2n} \to \mathbb{R}$ as
\begin{align*}
c([\zeta_1(t);s]) &:= \frac{s - f(\theta(\zeta_0(t),y(t)),y(t))}{g(\theta(\zeta_0(t),y(t)),y(t))}.
\end{align*}
Then, the ideal control \eqref{idealstate} is generated by
\begin{align}\label{ideal}
u(t) = c([\zeta_1(t); f_r(y(t))]).
\end{align}
For later use, we also define $\mathcal{C}$ as the set of all possible $[\zeta_1(t); f_r(y(t))]$. 
It is noted that $\mathcal{C} \subset \mathcal{E}$ by definition.

The ideal control \eqref{idealstate}, implemented as \eqref{ideal}, is an output feedback control in the sense that it uses input/output data only.
However, there is still a problem that the knowledge of functions $f(\cdot,\cdot)$, $g(\cdot,\cdot)$, and $\theta(\cdot,\cdot)$ are needed for constructing \eqref{ideal}. 
We solve this problem by applying GPR to identify the function $c(\cdot)$ itself.
This idea is feasible by treating $[\zeta_1(t);y(t+1)]$ as input data to the function $c$, and $u(t)$ as output data. 
This is because
$$u(t) = c([\zeta_1(t);y(t+1)])$$
from \eqref{SYSTEM1}.

To perform the proposal, we first collect input/output data{\footnote{The subscript $\dd$ is used for the sample data collected from the system during some experiment.}} of system \eqref{SYSTEM} as
\begin{align}\label{DATASET}
    \{(u_\dd(t),y_\dd(t))\}_{t=1}^{{N}}
\end{align}
where ${N}>n$ is the total number of input/output data.
Then we rearrange the data as the training input
\begin{align*}
    \xi_\dd(t+n-1) &= [\zeta_{1\dd}(t+n-1); y_\dd(t+n)] \\
    &= [ {y_\dd}_{[t,t+n-2]};{u_\dd}_{[t,t+n-2]};   {y_\dd}_{[t+n-1,t+n]}]
\end{align*}
and the training output
$$u_\dd {(t+n-1)},$$
yielding the training dataset:
\begin{align}\label{TD}
    {\mathcal{D}}:=\left\{ (\xi_\dd(t+n-1), u_\dd (t+n-1) )
    \right\}_{t=1}^{N-n}.
\end{align}

\begin{rem}\label{rem:multiple}
It may be difficult to collect sufficiently long ($N \gg n$) input/output sequences as in \eqref{DATASET} if the system is unstable.
In this case, let input/output data collected in the $i$-th experiment be
\begin{align*}
\{(u^i_\dd(t),y^i_\dd(t))\}_{t=1}^{{N}^i},
\end{align*} 
where ${N}^i>n$.
The training input and output samples are rearranged as 
\begin{align*}
\xi^i_\dd(t+n-1) = [ {y^i_\dd}_{[t,t+n-2]};{u^i_\dd}_{[t,t+n-2]};   {y^i_\dd}_{[t+n-1,t+n]}]
\end{align*}
and
$$u^i_\dd {(t+n-1)}$$
for $t=1,\ldots,N^i-n$, respectively.
The training dataset for each $i$-th experiment is defined as
\begin{align*}
    {\mathcal{D}^i}:=\left\{
    (\xi^i_\dd(t+n-1), u_\dd^i (t+n-1) )
    \right\}_{t=1}^{{N}^i-n}.
\end{align*}
By combining all ${\mathcal{D}^i}$, we can still use 
\begin{align*}
    {\mathcal{D}}=\bigcup_{i}{\mathcal{D}^i}
\end{align*} for identifying the function $c(\cdot)$ even if the data of each experiment is obtained from different initial conditions. 
In this case, the total number of input/output data becomes $N=\sum_{i} N^i$.
$\hfill\Box$
\end{rem}

A Gaussian process (GP) is fully specified by a mean function $m: {\mathcal{E}} \to \mathbb{R}$ and a covariance function $k: {\mathcal{E}} \times {\mathcal{E}} \to \mathbb{R}$.
The GP that we use for identifying the function $c(\cdot)$ employs the zero function for the mean function, and the squared exponential (SE) kernel for the covariance function:
\begin{align}\label{SEkernel}
    k({\xi},{\xi}') = \sigma_{f}^2 \text{exp}\left( -\frac12 (\xi-\xi')^T L^{-1} (\xi-\xi') \right)
\end{align}
in which, $\sigma_{f}$ and $L = {\rm diag}(l_{1}^2, \ldots, l_{2n}^2)$ are hyperparameters.
We set the hyperparameters using the marginal likelihood optimization according to Bayesian principles \cite[Chapter~5]{GPML}.

By using the training dataset ${\mathcal{D}}$ in \eqref{TD}, the GP yields the posterior mean and variance functions for a test input $\xi \in \mathcal{E}$
\begin{align}
    \mu_{\mathcal{D}}(\xi) &:= \mathbf{k}^T ({\xi}){\mathbf{K}}^{-1} \mathbf{u}, \label{mu} \\
    \sigma_{\mathcal{D}}(\xi) &:= k ({\xi},{\xi}) - \mathbf{k}^T ({\xi}) {\mathbf{K}}^{-1} \mathbf{k} ({\xi}), \label{sig}
\end{align}
respectively, where 
\begin{align*}
\mathbf{u} &:= [{u}_\dd{(n)};\cdots;{u}_\dd{({N}-1)}],\\
\mathbf{k} ({\xi})&:= [k(\xi_\dd{(n)},{\xi});\cdots;k(\xi_\dd{({N}-1)},{\xi})],\\
\mathbf{K} &:= 
\resizebox{.43\textwidth}{!}{$
\begin{bmatrix}
k(\xi_\dd{(n)},\xi_\dd{(n)}) & \cdots & k(\xi_\dd{(n)},\xi_\dd{({N}-1)})\\
\vdots & \ddots & \vdots \\
k(\xi_\dd{({N}-1)},\xi_\dd{(n)}) & \cdots & k(\xi_\dd{({N}-1)},\xi_\dd{({N}-1)})
\end{bmatrix}.
$}
\end{align*}
The posterior mean and variance functions are not well-defined if, for example, some elements in the set of training inputs $\{ \xi_\dd(t+n-1)\}_{t=1}^{N-n}$ are identical as mentioned in \cite[Remark~3.3]{Kanagawa18}, which hardly occurs in practice.

It is noted that the posterior mean function $\mu_{\mathcal{D}}(\cdot)$ is in fact the estimation result of the function $c(\cdot)$ which is obtained by only the input/output data of the system \eqref{SYSTEM}. On the other hand, the posterior variance function $\sigma_{\mathcal{D}}(\cdot)$ indicates the confidence of the estimation.

Finally, we construct the MR-GPR controller by using the mean function $\mu_{\mathcal{D}}$ in \eqref{mu} as 
\begin{align}\label{controller}
\begin{split}
u(t) &= \mu_{\mathcal{D}} ([\zeta_1(t); f_r(y(t))]) \\
&= \mu_{\mathcal{D}} ([y_{[t-n+1,t-1]}; u_{[t-n+1,t-1]}; y(t); f_r(y(t))])
\end{split}
\end{align}
which is an output feedback controller.
The following theorem shows the convergence of the closed-loop system with the MR-GPR controller \eqref{controller} under a boundedness assumption of the input gain.

\begin{asm}\label{ASM:BD}
There exists $\bar{g}>0$ such that $\lvert g(z,y) \rvert \leq \bar{g}$ for all $(z,y) \in \mathbb{R}^n$. $\hfill\Box$
\end{asm}

While we assume the input gain function $g(\cdot,\cdot)$ to be bounded, if we consider the case where $z(t)$ and $y(t)$ stay in some compact sets, then it is seen that the boundedness directly follows from the smoothness of the input gain function.

\begin{thm}\label{THM1}
Under Assumptions \ref{ASM:1}, \ref{ASM:OBS}, and \ref{ASM:BD}, there exists a class-$\KK$ function $\gamma$ such that, if there exists a dataset ${\mathcal{D}}$ so that
\begin{equation}\label{eq:thmcond}
|\mu_{\mathcal{D}}([\zeta_1;s]) - c([\zeta_1;s])| < \delta,\quad \forall [\zeta_1;s] \in \mathcal{C}
\end{equation}
for a given $\delta > 0$, then, the closed-loop system \eqref{SYSTEM} with the MR-GPR controller \eqref{controller} guarantees
$$\limsup_{t \to \infty} \|[y(t);z(t)]\| < \gamma(\delta).$$
\end{thm}

\medskip

\begin{proof}
For notational simplicity, let $y_t$ imply $y(t)$ in this proof.
Applying \eqref{controller} to \eqref{SYSTEM1}, we have
\begin{align*}
y_{t+1} &= f(z_t,y_t) + g(z_t,y_t) \mu_{\mathcal{D}}([\zeta_{1,t};f_r(y_t)]) \\
&= f(z_t,y_t) + g(z_t,y_t) c([\zeta_{1,t};f_r(y_t)]) \\
&\qquad + g(z_t,y_t) \{\mu_{\mathcal{D}}([\zeta_{1,t};f_r(y_t)]) - c([\zeta_{1,t};f_r(y_t)])\} \\
&= f_r(y_t) + e_t
\end{align*}
where 
\begin{align} \label{error}
e_t := g(z_t,y_t) \{\mu_{\mathcal{D}}([\zeta_{1,t};f_r(y_t)]) - c([\zeta_{1,t};f_r(y_t)])\}.
\end{align} 
By the assumption,
$$\lvert e_t \rvert \leq \bar{g}\delta, \qquad \forall t$$
and by the input-to-state stability assumption of the reference model, there are a class-$\mathcal{KL}$ function $\beta_y$ and a class-$\KK$ function $\gamma_y$ such that
$$|y_t| \le \beta_y(|y_0|,t) + \gamma_y( \bar g \delta ).$$

Also from the input-to-state stability of \eqref{SYSTEM2} in Assumption \ref{ASM:1} (b), there exists a class-$\mathcal{KL}$ function $\beta_z$ and a class-$\mathcal{K}$ function $\gamma_z$ such that
\begin{align*}
\|z_t\| &\le \beta_z(\|z_0\|,t) + \gamma_z(|y_t|).
\end{align*}
Therefore, 
\begin{align*}
\limsup_{t\to\infty} |y_t| &\le \gamma_y( \bar g \delta ) \\
\limsup_{t\to\infty} \|z_t\| &\le \gamma_z( \gamma_y(\bar g \delta) )
\end{align*}
so that the function $\gamma$ that completes the proof can be constructed.
\end{proof}

\begin{rem}
We identify the smooth function 
$c(\cdot)$ as $\mu_{\mathcal{D}}(\cdot)$ by the GP with SE kernel. In fact, the posterior variance function $\sigma_{\mathcal{D}}(\cdot)$ in \eqref{sig} can be utilized to measure 
how much the function $\mu_{\mathcal{D}}(\cdot)$, the identification result, differs from the function $c(\cdot)$.
Specifically, if the function $c(\cdot)$ belongs to reproducing kernel Hilbert space generated by the kernel $k$ in \eqref{SEkernel}, then
\begin{align*}
    |\mu_{\mathcal{D}}(\xi)-c(\xi)| \le \beta \sqrt{\sigma_{\mathcal{D}} (\xi)}, \quad \forall \xi\in \mathcal{E}
\end{align*}
for some positive $\beta$ (see \cite[Corollary~3.11]{Kanagawa18} for details).
In addition, \cite[Corollary~3.2]{Lederer19} presents a certain method for data collection, with which it is possible to make the upper bound (the function of the posterior variance) arbitrarily small by using a sufficiently large number of data $N$. 
With the help of these facts, we can compose a dataset $\mathcal{D}$ that satisfies the sufficient condition \eqref{eq:thmcond} in Theorem \ref{THM1} for a given $\delta>0$.
$\hfill\Box$
\end{rem}

\begin{rem}
In order to initiate the output feedback controller \eqref{controller} at time $t=0$, information of $y_{[-n+1,-1]}$ and $u_{[-n+1,-1]}$ is needed.
If the system is initially at rest or at the steady-state, then the information is easy to obtain, but this may not be the typical situation.
Instead, one may apply arbitrary inputs for the initial $(n-1)$ steps.
In fact, the information about $(n-1)$-long input/output sequences is necessary to figure out the information of internal state, which is reminiscent to the classical output feedback controls, in which, the state-feedback control does not have much meaning until a dynamic observer estimates the plant's state.
$\hfill\Box$
\end{rem}

\section{Illustrative Example}\label{sec:ex}
In this section, an illustrative example is presented to describe the utility of the proposed data-driven controller.

Consider the following SISO system
\begin{subequations}\label{ex:toy}
\begin{align}
    y(t+1) &= y^2(t)+z(t)+u(t) \label{ex:y}\\
    z(t+1) &= 0.5\sin(y(t))z(t) \label{ex:zerodyn},  
\end{align}
\end{subequations}
where $u,y,z \in \mathbb{R}$. 
We assume that the system dimension $n=2$ and the global relative degree one are known (Assumption \ref{ASM:1} (a)).
Since the internal dynamics \eqref{ex:zerodyn} is input-to-state stable from $y$ to $z$, the system \eqref{ex:toy} also satisfies Assumption \ref{ASM:1} (b).
Furthermore, Assumption \ref{ASM:OBS} is satisfied by the fact that $z(t)$ is uniquely determined by
\begin{align*}
    z(t) &= y(t+1)-y^2(t)-u(t) \\
    &= \mathcal{O}([y_{[t,t+1]};u(t)]).
\end{align*}
Therefore, we obtain
\begin{align*}
    z(t) &= 0.5\sin(y(t-1))z(t-1) \\
    &= 0.5\sin(y(t-1))\left(y(t)-y^2(t-1)-u(t-1)\right) \\
    &= \theta\left(\zeta_0(t),y(t)\right)
\end{align*}
as in Lemma \ref{LEM:z}.
Noting that
\begin{align*}
    c([\zeta_1(t);y(t+1)]) = y(t+1) - y^2(t) - \theta\left(\zeta_0(t),y(t)\right),
\end{align*}
we compose the training data ${\mathcal{D}}=\bigcup_{i=1}^{T} \mathcal{D}^i$,
where the data $\mathcal{D}^i$ is collected in an experiment with random initial condition $y_\dd(0), z_\dd(0) \in [-1.2,1.2]$ and random input $u_\dd(t) \in [-1.2,1.2]$ for $N^i=5$ time steps for all $i=1,\ldots,T$. In the $i$-th experiment, as in Remark \ref{rem:multiple}, we obtain the data
$$\xi_\dd(t+1)=\left[ y_\dd(t);u_\dd(t);y_\dd(t+1); y_\dd(t+2)\right]$$ 
which is used as a training input and
$$u_\dd(t+1)$$ 
which is considered as a training output for $t=1,2,3$.
Using the training data ${\mathcal{D}}$, we set the hyperparameters in \eqref{SEkernel} by optimizing the marginal likelihood through GPML toolbox \cite{gpmltoolbox}.
Finally, we take a stable reference model as
$$y_r(t+1) = f_r(y_r(t)) = -0.4 y_r(t)$$
which guarantees input-to-state stability for $$y_r(t+1)=-0.4 y_r(t) + \eta(t).$$
Then, the proposed output feedback controller becomes
\begin{align}
\begin{split}
u(t) &= \mu_{\mathcal{D}}([\zeta_1(t); -0.4 y(t)]) \\
&= \mu_{\mathcal{D}}([y(t-1); u(t-1); y(t); -0.4 y(t)]).
\end{split}
\end{align}

Figs.~\ref{FIG:toy_epi20} and \ref{FIG:toy_epi2000} show the output of the closed-loop system with the proposed controller designed by the training data of $T=20$ and $2000$ experiments from different initial conditions, respectively, compared to the one with the ideal controller. 
We set the initial conditions of each system as 
\begin{align*}
  &\left(y(0),z(0)\right) \\
  &\quad\quad\in \{(1.1,1.1), (1.1,-1.1), (-1.1,1.1), (-1.1,-1.1)\}  
\end{align*}
in both Figs.~\ref{FIG:toy_epi20} and \ref{FIG:toy_epi2000}.
In all cases, zero input $u(0)=0$ is used at the very first step of control for applying the MR-GPR controller.
It is observed that in both Figs.~\ref{FIG:toy_epi20} and \ref{FIG:toy_epi2000}, the MR-GPR controller asymptotically stabilizes all systems that have different initial conditions. Also, the MR-GPR controller designed with more data in Fig.~\ref{FIG:toy_epi2000} shows better performance than the one designed with less data in Fig.~\ref{FIG:toy_epi20}.
\begin{figure}[t]
\centering
\includegraphics[width=\columnwidth]{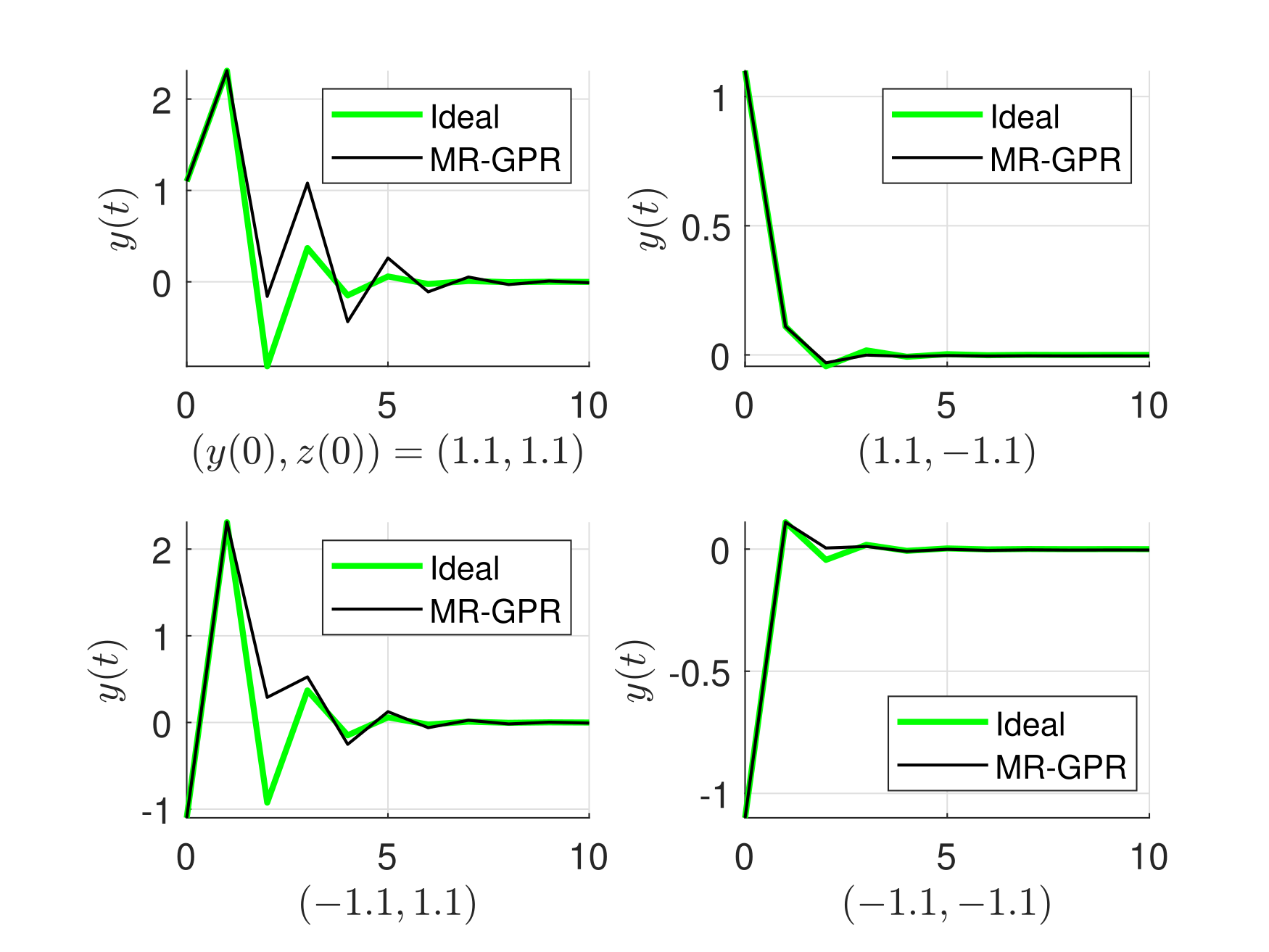}
\caption{Output trajectories of the system \eqref{ex:toy} with ideal 
controller $c$ (green line) and MR-GPR controller $\mu_\mathcal{D}$ (black line) designed by the data of $T=20$ experiments from different initial conditions.\label{FIG:toy_epi20}}
\end{figure}
\begin{figure}[t]
\centering
\includegraphics[width=\columnwidth]{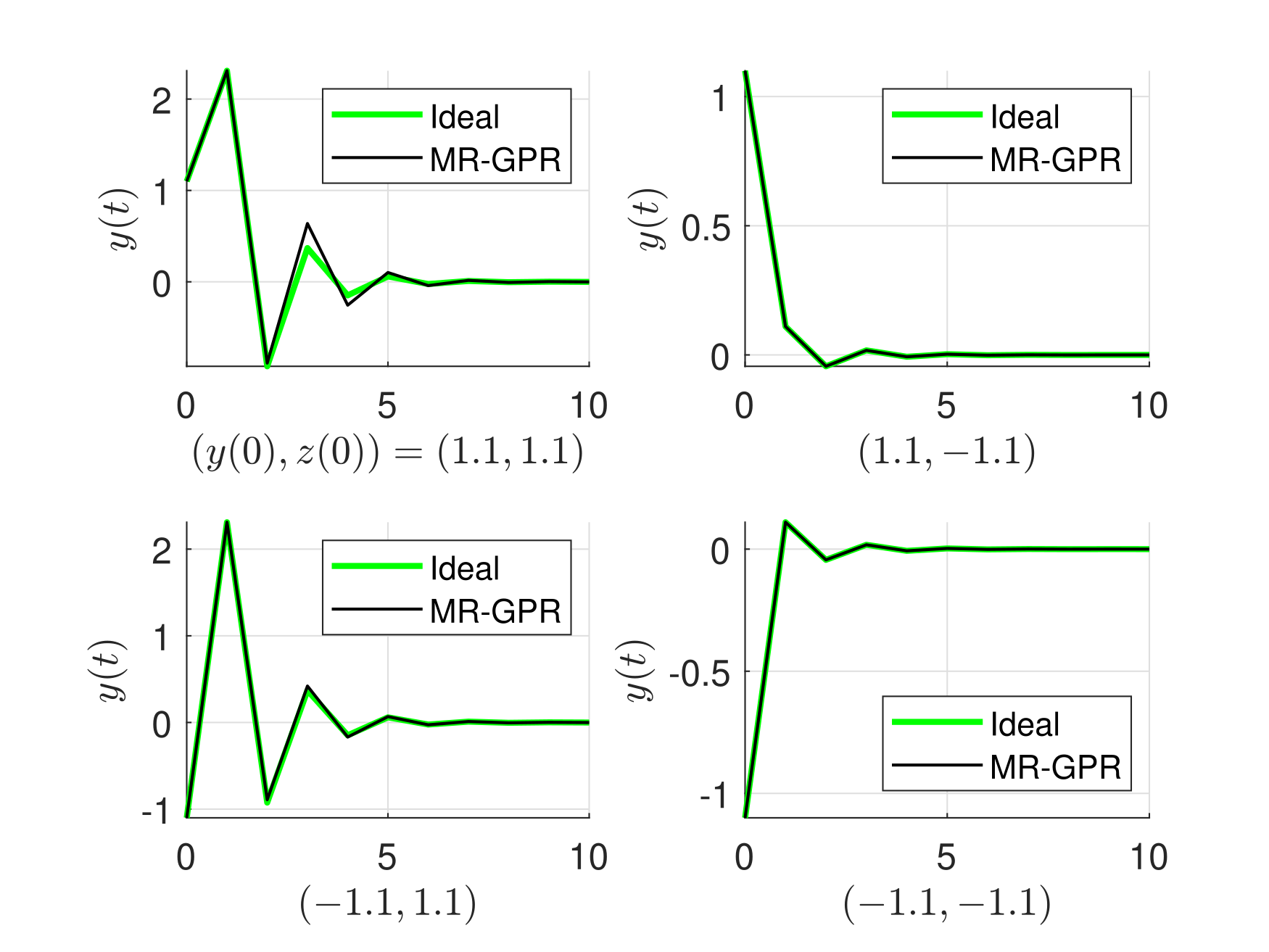}
\caption{Output trajectories of the system \eqref{ex:toy} with ideal 
controller $c$ (green line) and MR-GPR controller $\mu_\mathcal{D}$ (black line) designed by the data of $T=2000$ experiments from different initial conditions.\label{FIG:toy_epi2000}}
\end{figure}
\begin{figure}[t]
\centering
\includegraphics[width=\columnwidth]{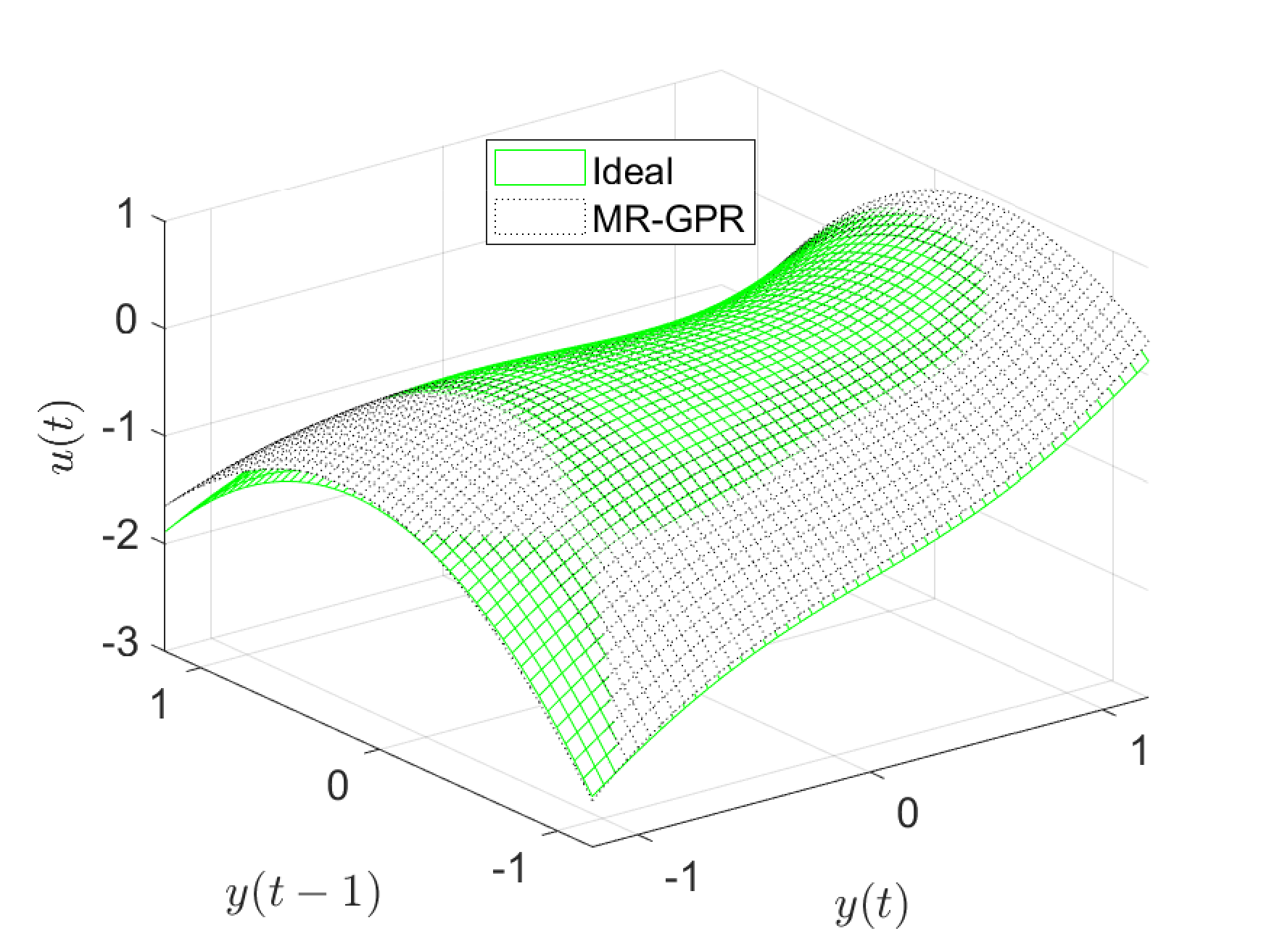}
\caption{Function values of ideal controller $c$ (green mesh) and MR-GPR controller $\mu_{\mathcal{D}}$ (black dotted mesh) designed by the data of $T=20$ experiments. \label{FIG:toy_mesh_20epi}}
\end{figure}
\begin{figure}[t]
\centering
\includegraphics[width=\columnwidth]{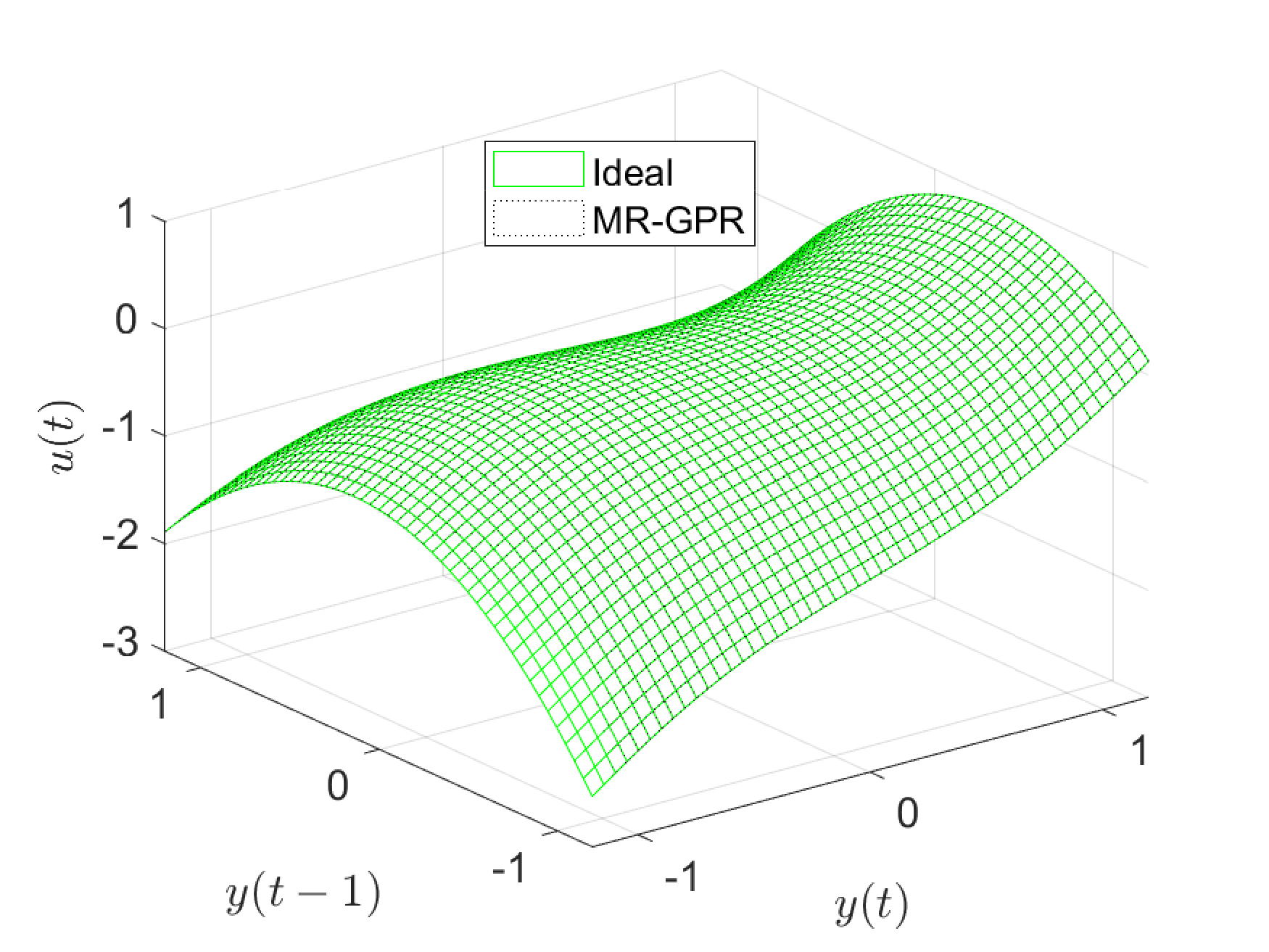}
\caption{Function values of ideal controller $c$ (green mesh) and MR-GPR controller $\mu_{\mathcal{D}}$ (black dotted mesh) designed by the data of $T=2000$ experiments. \label{FIG:toy_mesh_2000epi}}
\end{figure}
\begin{figure}[t]
\centering
\includegraphics[width=\columnwidth]{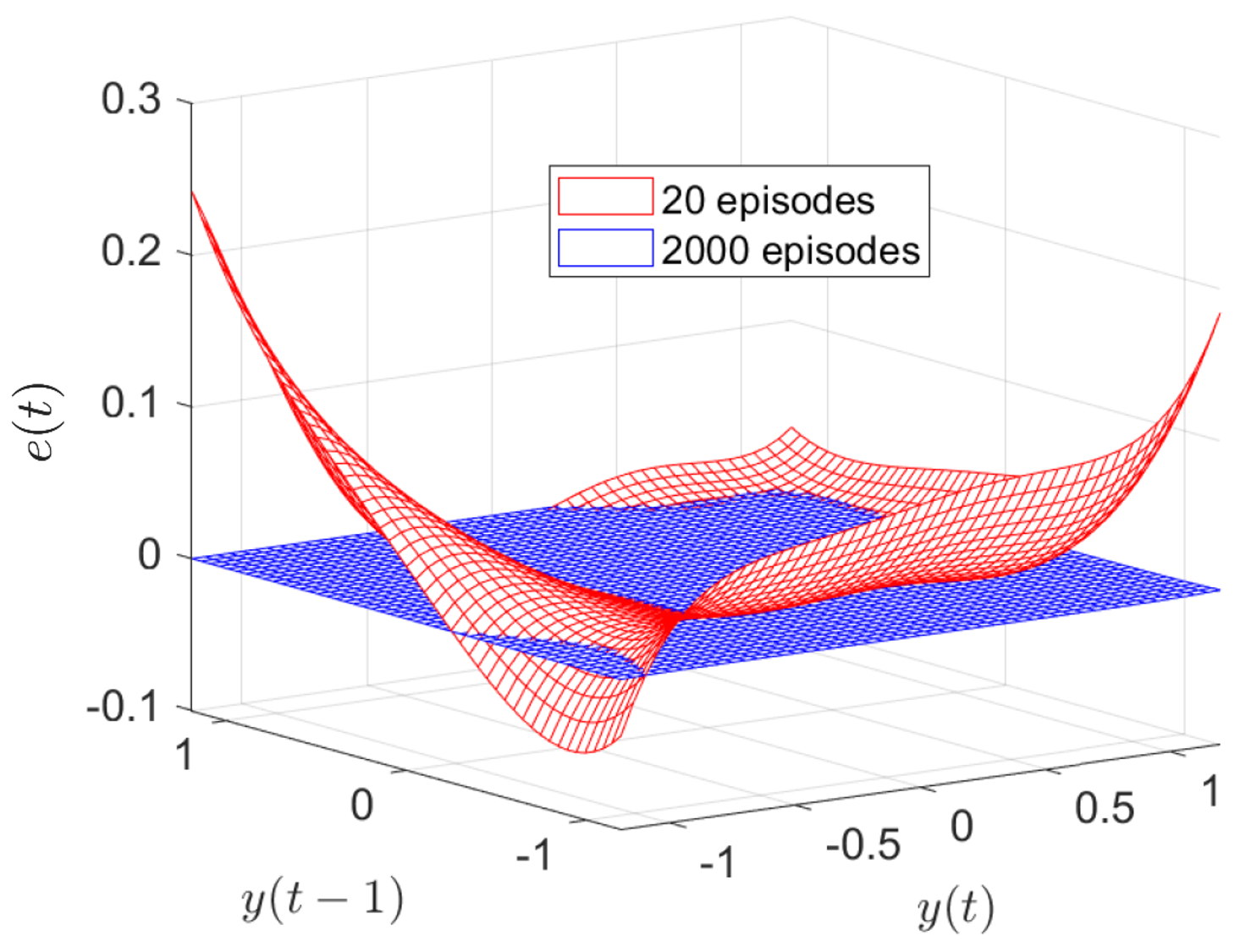}
\caption{Error of function values between ideal controller $c$ and MR-GPR controller $\mu_{\mathcal{D}}$ designed by the data of $T=20$ experiments (red mesh) and $T=2000$ experiments (blue mesh)
\label{FIG:toy_error}}
\end{figure}

On the other hand, Figs.~\ref{FIG:toy_mesh_20epi} and \ref{FIG:toy_mesh_2000epi} depict function values of the MR-GPR controller $\mu_{\mathcal{D}}$ designed by using the data of $T=20$ and $2000$ experiments, respectively, compared to the ideal controller $c$. Although both functions $\mu_{\mathcal{D}}$ and $c$ need an input
\begin{align*}
[\zeta_1(t); -0.4y(t)] =[y(t-1);u(t-1);y(t);-0.4y(t)]
\end{align*}
to be evaluated, we fix the value $u(t-1) = 0.2$ and evaluate both functions by sweeping $y(t-1)$ and $y(t)$ in $[-1.2,1.2]$.
It is seen that the proposed controller $\mu_{\mathcal{D}}$ sufficiently well approximates the ideal controller $c$ throughout the entire domain when $2000$ experiments of data are used in Fig~\ref{FIG:toy_mesh_2000epi}, while the approximation reveals some error particularly at evaluation points which are far from $(0,0)$ when $20$ experiments of data are used Fig~\ref{FIG:toy_mesh_20epi}. 
This is also verified in Fig.~\ref{FIG:toy_error} that plots the error $e_t$ of \eqref{error}, which is a function of $[\zeta_1(t); -0.4y(t)]$.

\section{Conclusion}\label{sec:conc}

In this paper, we proposed the MR-GPR controller, which is the data-driven output feedback controller, for SISO nonlinear discrete-time control-affine systems with relative degree one and of minimum phase. 
The design was performed by using the GPR, trained only by input/output data of the system. 
It is worthy to emphasize that the GPR was utilized not for system identification but for controller design itself.
It was shown that the control performance improves as more data are available for training, which was demonstrated by an illustrative example using simulations.

While we present a new concept of MR-GPR, the class of applicable systems is still restricted, and more study is necessary to extend the applicable system class.
Future research topics include regression of residual nonlinearity in the inversion, extension to multi-input multi-output case, consideration of measurement noises, handling of non-minimum phase systems by a feedback of estimated internal states, and application of the proposed idea to other regression methods \cite{sindy} and \cite{ml}.


\end{document}